\documentclass[11pt]{article}

\usepackage[english]{babel}
\usepackage[utf8x]{inputenc}

\usepackage{booktabs}
\usepackage{tabu}
\usepackage[T1]{fontenc}

\usepackage[a4paper,top=2.5cm,bottom=2.5cm,left=2.5cm,right=2.5cm,marginparwidth=1.75cm]{geometry}

\usepackage{sansmathfonts}
\usepackage[sfdefault]{biolinum}
\usepackage[T1]{fontenc}
\usepackage[usenames,dvipsnames,svgnames]{xcolor}

\usepackage{amsmath}
\usepackage{amsthm}
\usepackage{amssymb}
\usepackage{graphicx}
\usepackage{amsthm}
\usepackage[noadjust]{cite}
\usepackage{mathtools}
\usepackage{graphicx}
\usepackage{array}
\usepackage{xspace}
\usepackage{float}
\usepackage{wrapfig}
\usepackage{thmtools, thm-restate}
\usepackage{enumitem}
\usepackage{subcaption}

\usepackage{soul}
\soulregister\cite7
\soulregister\ref7
\soulregister\cite7
\soulregister\cref7
\colorlet{soulblue}{cyan!30}
\sethlcolor{soulblue}

\usepackage[boxruled, nofillcomment, lined]{algorithm2e}

\RequirePackage{hyperref}
\hypersetup{colorlinks=true, urlcolor=Blue, citecolor=Green, linkcolor=BrickRed, breaklinks, unicode}

\usepackage[colorinlistoftodos]{todonotes}
\usepackage[capitalise, noabbrev]{cleveref}

\theoremstyle{plain}
\newtheorem{theorem}{Theorem}
\newtheorem{lemma}{Lemma}

\newtheorem{claim}{Claim}

\theoremstyle{definition}
\newtheorem{definition}{Definition}

\newtheorem{observation}{Observation}

\theoremstyle{remark}

\newcommand{\CONGEST}{\mathsf{CONGEST}}
\newcommand{\Sin}{S_{in}}
\newcommand{\Sout}{S_{out}}
\newcommand{\wF}{w_{\!_F}}
\newcommand{\wV}{w_{\!_V}}

\newcommand{\wA}{w_{\!_A}}

\newcommand{\depth}{\mathsf{depth}}
\newcommand{\parent}{\mathsf{parent}}

\newcommand{\Poly}{\text{poly}}

\title{A Simple Distributed Deterministic Planar Separator}

\author{%
 Yaseen Abd-Elhaleem\thanks{{\em Department of Computer Science, University of Haifa. Supported in part by the Israel Science Foundation grant No. 810/21 and No. 2829/25.}}\\
 \texttt{\href{mailto:yaseenuniacc@gmail.com}{yaseenuniacc@gmail.com}}\\
 \and Michal Dory\thanks{{\em Department of Computer Science, University of Haifa. Supported in part by the Israel Science Foundation grant No. 2829/25.}}\\
 \texttt{\href{mailto:mdory@ds.haifa.ac.il}{mdory@ds.haifa.ac.il}}\\
 \and Oren Weimann\thanks{{\em Department of Computer Science, University of Haifa. Supported in part by the Israel Science Foundation grant No. 810/21.}}\\
 \texttt{\href{mailto:oren@cs.haifa.ac.il}{oren@cs.haifa.ac.il}}
}

\date{}

\begin{document}
\maketitle

\begin{abstract}
A balanced separator of a graph $G$ is a set of vertices whose removal disconnects the graph into connected components that are a constant factor smaller than $G$.
Lipton and Tarjan [FOCS'77] famously proved that every planar graph admits a balanced separator of size $O(\sqrt{n})$, as well as a balanced separator of size $O(D)$ that is a simple path (where $D$ is the graph's diameter).
In the centralized setting,  these separators can both be found in linear $O(n)$ time. 
In the distributed setting, since the diameter $D$ is a trivial universal lower bound for the number of rounds required to solve many  optimization problems, separators of size $O(D)$ are preferable over those of size  $O(\sqrt{n})$. 
    
It was not until [Ghaffari, Parter DISC'17] that an algorithm was devised to compute such an $O(D)$-size separator distributively in $\tilde O(D)$ rounds\footnote{The $\tilde{O}(\cdot)$ notation is used to omit $\Poly \log n$ factors.}, by adapting the Lipton-Tarjan algorithm to the distributed model. Since then, this algorithm was used in several distributed algorithms for planar graphs, e.g., [Ghaffari, Parter DISC'17], [Li, Parter STOC'19], [Abd-Elhaleem, Dory, Parter and Weimann PODC'25]. However, the algorithm is randomized, deeming the algorithms that use it to be randomized as well. Obtaining a deterministic algorithm remained an interesting open question until very recently, when a (complex) deterministic separator algorithm was given by   [Jauregui, Montealegre and Rapaport  PODC'25].

In this paper, we present a much simpler deterministic separator algorithm with the same (near-optimal)  $\tilde O(D)$-round complexity. 
While previous works devised either complicated or randomized ways of transferring weights from vertices of $G$ to faces of $G$, we show that a straightforward way also works: Each vertex simply transfers its weight to one arbitrary face it belongs to. That's it!
    
We note that a deterministic separator algorithm directly derandomizes the state-of-the-art distributed algorithms for classical problems on planar graphs such as single-source shortest-paths, maximum-flow, directed global min-cut, and reachability. 
\end{abstract}

\thispagestyle{empty}
\setcounter{page}{0}

\newpage

\section{Introduction}
A $c$-balanced separator of a vertex-weighted graph $G$ is a set of vertices whose removal disconnects the graph into connected components, each weighing at most a $c$ fraction of $G$'s total weight.
In their classical paper from the 70's~\cite{LiptonT80_sep}, Lipton and Tarjan showed that every planar graph admits a balanced separator of size $O(\sqrt{n})$. This has been used in numerous divide-and-conquer centralized algorithms. Their proof first shows that given any spanning tree $T$ of $G$, there is a balanced separator of $G$ consisting of a path in $T$. If $G$ is {\em triangulated}, then this separator is a {\em fundamental cycle} (a cycle consisting of the path in $T$ plus an edge $e$ of $G$ connecting the path's endpoints). 
Taking $T$ to be a BFS tree produces a separator of size $O(D)$, where $D$ is the graph's hop-diameter. 
 In the distributed setting, since $D$ is a trivial lower bound for the number of rounds required to solve many  optimization problems, separators of size $O(D)$ are preferable over those of size $O(\sqrt{n})$. 
 
Lipton and Tarjan reduce the problem of computing a balanced separator in $G$ (using a spanning tree $T$) to computing a {\em balanced cut} in the {\em dual tree} $T^*$. It 
relies on the cut-cycle duality in planar graphs (see \cref{sec: preliminaries} for definitions of duality), and on the fact that one can transfer the weights from vertices of $G$ to faces of $G$ (i.e., to nodes of the dual graph $G^*$). 
In order for $T^*$ to have a balanced cut, it is crucial that $T^*$ has a constant degree. In the centralized setting, this can easily be achieved by triangulating $G$, i.e., adding artificial edges to $G$ to make every face of size 3 (thus making every node in $T^*$ to be of degree at most 3).
In the distributed setting, however, one cannot afford to triangulate $G$ (as artificial edges cannot be communicated on). Instead, Ghaffari and Parter~\cite{GhaffariP17_dfs} gave a distributed implementation of Lipton-Tarjan that circumvents this using a randomized procedure that approximates the face-weights assigned by Lipton-Tarjan.
Since then, their separator algorithm has been used in several state-of-the-art distributed algorithms for planar graphs, including depth-first search, single-source shortest-paths, maximum flow, routing, and diameter~\cite{AbdElhaleemDPW25_flow,LiP19_diameter,GhaffariP17_dfs,routing}.
However, since the distributed separator algorithm was randomized, all algorithms that used it were randomized as well. In fact, in most of them, finding the separator is the only randomized component. 
Very recently in PODC 2025, using a different approach, Jauregui, Montealegre, and Rapaport~\cite{deterministic_sep} provided a {\em deterministic} distributed separator algorithm. 
Their algorithm avoids computations on the dual graph by instead considering a collection of triangulation edges. This, however leads to a complex analysis, as we discuss in \cref{sec:overview}.

\medskip
\noindent
{\bf Our result.} 
The standard way to compute a balanced separator in the primal graph (with respect to vertex-weights) is to translate it to the problem of computing a balanced cut (a cut whose sides are roughly of the same weight) in the dual graph. To do so, one needs to assign weights to the dual nodes (faces of the primal graph) such that a balanced cut in the dual graph translates to a balanced separator in the primal graph.
A natural way to transfer weights from vertices to faces is to set  a face weight to be the sum of weights of its vertices. This however does not work, since a vertex can belong to an arbitrary number of faces, which leads to overcounting.
This issue was the source of~\cite{GhaffariP17_dfs} resorting to randomization and of~\cite{deterministic_sep} being complex. We give a very simple and natural alternative: Each vertex transfers its weight to one arbitrary face it belongs to. That's it! Since each vertex transfers its weight to one face it is counted exactly once which overcomes the problem of overcounting, and we show that a balanced cut with respect to this weight assignment translates to a balanced separator.
Our result is summarized by the following theorem. 

\begin{restatable}{theorem}{DistributedDeterministicSeparator}
\label{thm: det_separator}
    Let $G$ be an embedded planar graph of hop-diameter $D$,  $T$ be a spanning tree of $G$, and $w(\cdot)$ be a weight assignment to $G$'s vertices s.t. no vertex weighs more than $\frac{1}{12}$
    fraction of the total weight of $G$. There is a deterministic $\tilde{O}(D)$-round  distributed algorithm that finds a $u$-to-$v$ path  $P$ in $T$ that is a  $\frac{3}{4}$-balanced  separator of $G$. Adding the edge $e=(u,v)$ to $G$ (if it does not already exist), closes a fundamental cycle $P\cup\{e\}$ in $T$.
    Setting $T$ to be a BFS tree produces a separator of size $O(D)$. 
\end{restatable}

We note that a deterministic separator algorithm directly derandomizes the state-of-the-art distributed algorithms for classical problems on planar graphs such as depth-first search (DFS),  
single-source shortest-paths, maximum-flow, directed global min-cut, and reachability. This already follows from~\cite{deterministic_sep}. However, in~\cite{deterministic_sep} they only mention the application of DFS. We elaborate on the above other applications in \cref{sec:applications}. 

Our algorithm works with vertex-weighted graphs, where previous works \cite{GhaffariP17_dfs, deterministic_sep} focused on the unweighted case (i.e., the unit-weight case. In particular, note that no vertex weighs more than $1/n$ fraction of the total weight).
Finally, it is often useful for applications to compute a separator in multiple vertex-disjoint subgraphs of $G$ simultaneously. Our algorithm can be easily extended to support this as we show in \cref{section: distributed_implementation}.

\medskip
\noindent
{\bf Roadmap.}
In \cref{sec: preliminaries} we give the necessary preliminaries. In \cref{sec:overview} we overview related work and our technical contribution. Our deterministic separator algorithm is presented in \cref{sec: separator}, and its  distributed implementation in \cref{section: distributed_implementation}. The applications are  discussed in \cref{sec:applications}. Finally, 
Appendix~\ref{appendix: deferred_details_implementation} includes some deferred details from Section~\ref{section: distributed_implementation}.

\section{Preliminaries}
\label{sec: preliminaries}
\noindent
{\bf The $\mathsf{\mathbf{CONGEST}}$ model.}
We work in the standard distributed $\CONGEST$ model~\cite{peleg-book}. Initially, each vertex knows only its unique $O(\log n)$-bit ID and the IDs  of its neighbors. Communication occurs in synchronous rounds.
In each round, each vertex can send each neighbor a distinct $O(\log n)$-bit message.
When the vertices of $G$ are weighted, we assume that the weights are polynomially bounded integers.
Thus, the weight of a vertex can be transmitted in $O(1)$ rounds. This is a standard assumption in the $\CONGEST$ model.
Input and output are local, e.g. when rooting a spanning tree, we assume that all vertices know (as an input) the ID of the root $r$, and their incident edges in the tree. When the algorithm halts, each vertex knows the ID of its parent in the tree.

 \medskip
\noindent
{\bf Basic notation and graph theory.}
We denote by $G=(V,E)$ the vertex-weighted or face-weighted simple (connected) planar graph of communication. 
We denote by $D$ the hop-diameter of $G$ and by $F$ the set of faces of $G$ in a given planar embedding (defined next).
When the weights are assigned to vertices (resp. faces), we denote the weight function by $\wV(\cdot)$ (resp. $\wF(\cdot)$).
We denote by $w_{A}(G)$ the total weight of $G$ w.r.t. $w_{A}(\cdot)$ where $A\in\{V,  F\}$. 
When $A$ is clear from context we may omit it from the subscript and write $w(\cdot)$.

\medskip
\noindent
{\bf Planar embedding.}
The {\em geometric} planar embedding of a planar graph $G$ is a drawing of $G$ on a plane so that edges intersect only in vertices. In such embedding, there is a distinguished face $f_\infty$ (called the infinite face) surrounding the graph. When we consider a cycle $C$ in $G$, we refer to $C$'s {\em exterior} (resp. {\em interior}) as the side of $C$ that contains (resp. does not contain) $f_\infty$.
We denote by $C_{in}$ (resp. $C_{out}$) the interior (resp. exterior) of $C$, including $C$ itself. We denote the {\em strict} interior (resp. exterior) of $C$ (i.e., excluding $C$) by $C_{in}^-$ (resp. $C_{out}^-$). 

A {\em combinatorial} planar embedding of $G$ provides for each vertex $v\in G$, the local clockwise ordering of its incident edges, such that the ordering is consistent with some geometric planar embedding of $G$. See Chapter 3.4 of~\cite{KM_planarity_book} for more information.
Throughout, we assume that a combinatorial embedding of $G$ is known locally for each vertex. This is achieved in $\tilde{O}(D)$ rounds using the  \emph{deterministic} distributed planar embedding algorithm of Ghaffari and Haeupler~\cite{GhaffariH_embedding16}.

\begin{wrapfigure}{r}{0.25\linewidth}
  \includegraphics[width=1\linewidth]{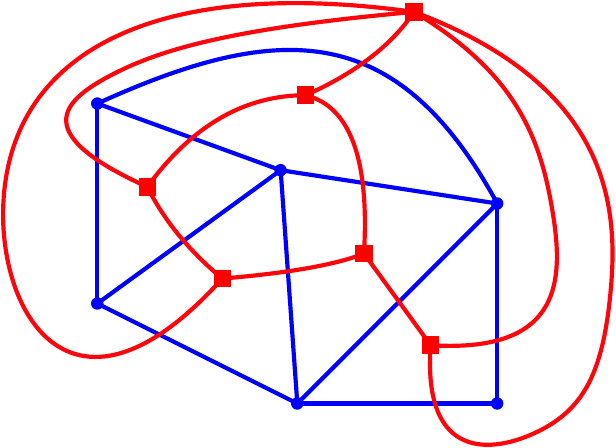}
 \end{wrapfigure}

\medskip
\noindent
{\bf Planar duality.}
\label{def: dual_graph}
The {\em dual} of a {\em primal} embedded planar graph $G$ is a planar graph $G^*$ whose nodes \footnote{For clarity, we will refer to faces of the primal graph $G$ as {\em nodes} (rather than vertices) of the dual graph $G^*$.} correspond to the faces of $G$.
For each edge $e\in G$ that belongs to two faces in $G$, there is an edge $e^*$ in $G^*$ between the faces' corresponding nodes (and a self-loop when the same face appears on both sides of an edge). We sometimes abuse notations and denote both $e$ and $e^*$ as $e$. 
See 
the figure for an illustration (the primal graph $G$ in blue and its dual graph $G^*$ in red).\\

\begin{lemma}[Tree and cotree, Chapter 4.5 of~\cite{KM_planarity_book}]
\label{fact: tree_cotree_decomposition}
     $T$ is a spanning tree of $G$  iff the {\em cotree} $T^*=E\setminus T$ is a spanning tree of $G^*$. 
\end{lemma}

\noindent {\bf Fundamental cuts and cycles.} We use the duality between cycles defined by $T$ and cuts defined by $T^*$: 
\begin{definition}[Fundamental cuts]
    Removing an edge $e\in T$ from $T$, breaks $T$ into two connected components. The set $\delta_G(T,e)$ of edges of $G$ that have one endpoint in each component is called the fundamental cut of $T$ with respect to $e$ in $G$.
\end{definition}  

\begin{definition}[Fundamental cycles]
    Adding an edge $e\in G\setminus T$ to $T$, closes one simple cycle $C_G(T,e)$ with $T$.
    The cycle $C_G(T,e)$ is called the fundamental cycle of $T$ with respect to $e$ in $G$.
\end{definition}  

\noindent When $G$ is clear from the context, we use $\delta(T,e)$ and $C(T,e)$ to denote $\delta_G(T,e)$ and  $C_G(T,e)$ respectively.

\begin{lemma}
    [Fundamental cut-cycle duality~\cite{KM_planarity_book}]
    \label{fact: fundamental_cycle_cut_duality}
   Let $T^*$ be a spanning tree of $G^*$ and $T=E\setminus T^*$ be its cotree. $T^*$ defines a fundamental cut in $G^*$ w.r.t. an edge $e^*$ iff $T$ defines a fundamental cycle in $G$ w.r.t. $e$. That is, $\delta_{G^*}(T^*,e^*)=C_G(T,e)$. 
   \end{lemma}

\noindent
{\bf Separators.}
Let $\wA(\cdot)$ be a weight function on $A\in \{V,F\}$.
\begin{definition}[$\alpha$-proper weights]
    $\wA(\cdot)$ is $\alpha$-proper for $\alpha\in(0,1)$ if $w_{A}(a)\leq \alpha \cdot w_{A}(G)$ for each $a\in A$. 
\end{definition}

\begin{definition}[Balanced separators]
   Let $c<1$ be a constant. 
   A $c$-balanced separator $S$ is a set of vertices s.t. the connected components of $G\setminus S$ are of size at most $c\cdot \wA(
    G)$.
    \end{definition}
    
\begin{definition}[Fundamental cycle separators]
   Let $c<1$ be a constant, and let $T$ be a spanning tree of $G$. A fundamental cycle separator of $G$ is a $c$-balanced separator $S$ which constitutes a fundamental cycle of $T$. I.e., there exists an edge $e\notin T$ s.t. $S=C(T,e)$, and $\wA(\Sin^-),\wA(\Sout^-)\leq c \cdot \wA(G)$. 
\end{definition}

\noindent
{\bf Balanced and critical nodes.} We use the following definitions of~\cite{GhaffariP17_dfs} that discusses rooted trees.  For a node $f\in T^*$, let $T^*_f$ denote the subtree of $T^*$ rooted at $f$. We denote by $w(T^*)$ the total weight of all nodes of $T^*$.

\begin{definition}[Balanced node]
    $f\in T^*$ is $(\alpha,\beta)$-balanced for constants $\alpha<\beta\leq 1$  if $\alpha \cdot w(T^*)\leq w(T^*_f) \leq \beta\cdot w(T^*)$. 
\end{definition}

\begin{definition}[Critical node]
        $f\in T^*$ is $(\alpha,\beta)$-critical for constants $\alpha<\beta\leq 1$ if $w(T^*_f)> \beta \cdot w(T^*)$, and $w(T^*_h)< \alpha \cdot w(T^*)$ for every child $h$ of $f$ in $T^*$. 
\end{definition}

\section{Technical Overview and Related Work}
\label{sec:overview}

Computing a balanced separator is a fundamental building block in many algorithms for planar graphs.
In distributed algorithms, it is desired that the separator is a path $P$ in some given spanning tree $T$ s.t. adding to $P$ an edge $e\not \in T$ (perhaps even $e\not \in G$) that connects $P$'s endpoints closes a fundamental cycle $C(T,e)$ in $G$ (or in $G\cup \{e\}$ if $e\not \in G$). Setting $T$ to be a BFS tree, we get a separator of size $O(D)$.
The goal is therefore to find such an edge $e$ where $C(T,e)$ encloses a constant $0<c<1$ fraction of $G$'s total weight of vertices. 
 Finding such an edge $e$ is not an easy task. 
 
However, if the weights were assigned to faces (rather than vertices), then the problem becomes easy, as it amounts to finding a balanced node $f$ in the dual tree $T^*$. Indeed, such a node $f$ naturally defines a fundamental cut $\delta_{G^*}(T^*,e^*)$ in $G^*$ (i.e.,  $e^*=(f,\parent(f))$) and hence a fundamental cycle $C(T,e)$ in $G$ (\cref{fact: fundamental_cycle_cut_duality}). See \cref{fig: cycle_cut}. Since $f$ is balanced, the weight of $f$'s  subtree $T^*_f$ in $G^*$ (the weight of all faces enclosed by $C(T,e)$ in $G$) is a constant fraction of $G$'s total weight of faces. Formally:

 \begin{figure}[htb]
\begin{minipage}[c]{0.26\textwidth}
    \includegraphics[width=\textwidth]{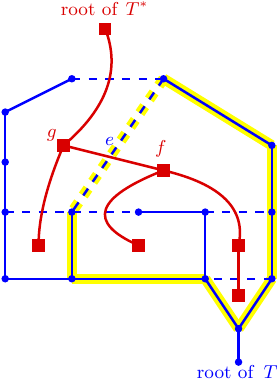}
  \end{minipage}\hfill
  \begin{minipage}[c]{0.6\textwidth}
    \caption{The primal graph $G$ is in blue. Solid blue edges are in $T$ and dashed blue edges are not in $T$. I.e., their duals (in red) are in $T^*$. Removing the dual edge $e^*=(g,f)$ from $T^*$ breaks $T^*$ into $T^*_f$ ($f$'s subtree in $T^*$) and $T^*\setminus T^*_f$. The fundamental cycle $C(T,e)$ (highlighted in yellow) encloses all faces that correspond to nodes of $T^*_f$, and does not enclose any face that corresponds to nodes of $T^*\setminus T^*_f$.
  \label{fig: cycle_cut}
 }
  \end{minipage}
\end{figure}

\begin{lemma}[See e.g., \cite{KM_planarity_book}]
\label{fact: face_weight_duality}
    Let $T$ be a spanning tree of $G$ and $e^*=(g,f)$ an edge of the cotree $T^*$ of $T$.
    If we root $T^*$ s.t. $g$ is closer to the root than $f$, then each node of $T^*_f$ maps to a face enclosed by $S=C(T,e)$.
    Analogously, each node of $T^*\setminus T^*_f$ maps to a face not enclosed by $S$. I.e., $w(T_f^*)=\wF(\Sin)$, and  $w(T^*\setminus T_f^*)=\wF(\Sout)$.
\end{lemma}

Alas, the weights in the input graph $G$ are assigned to vertices and not to faces. Therefore, the idea is to transfer the vertex-weights to face-weights s.t. a balanced separator w.r.t. these face-weights is also a balanced separator w.r.t. the original vertex-weights (perhaps with a different balance constant). There are two challenges with this approach:

\medskip \noindent {\bf Challenge I - Assigning Weights to Faces.} 
   The natural naive assignment that assigns each face a weight equal to the total weight of its vertices, does not work.
   This is because a vertex may belong to many faces.
   Thus, a (primal) subgraph might have a small total vertex-weight but a large total face-weight. 
    In $T^*$, this means that a subtree $T^*_f$ of some node $f$ might have a large weight, where in fact it corresponds to a cycle that encloses a  small-weight 
    subgraph, while the external subgraph has a large weight.  Such a subtree might be mistakenly chosen to define the separator. I.e., $f$ can be a balanced node in the dual tree, but it does not necessarily correspond to a balanced separator in the primal graph with respect to vertex-weights. 
    
\medskip \noindent {\bf Challenge II - Existence of a Balanced Node.} 
       Even if we overcome the first challenge,  since faces can be arbitrarily large, it is not guaranteed that a balanced node $f$ exists in $T^*$. It may be the case that all subtrees of $T^*$ either weigh too much, or weigh too little.  For example, when $T^*$ is a star.
        
\medskip We next describe how prior work dealt with the above challenges, and how we deal with them. 

\medskip
\noindent
{\bf Lipton-Tarjan~\cite{LiptonT80_sep}.} In the 70's, Lipton and Tarjan introduced the above approach in the centralized setting.  
To overcome Challenge I, each vertex of $G$ transfers its weight to an arbitrary face that it participates in.
Thus, for any cycle of $G$, the weight of its (non-strict) interior w.r.t. faces and w.r.t. vertices is similar (each vertex assigns its weight to exactly one face that it lies on). 
In the centralized setting, Challenge II
is easily dealt with by {\em triangulating} $G$. I.e., augmenting $G$ with artificial edges to make each of its faces consist of exactly three edges. 
In the augmented graph, it is not hard to show that $T^*$ always admits a balanced node $f$. 
This is because now the dual graph $G^*$ is 3-regular, thus $T^*$ is a binary tree, and the weight assignment to faces is $3\alpha$-proper if the weight to vertices was originally $\alpha$-proper. 
In the distributed setting however, triangulation is not an option, since it may require adding a linear number of artificial edges which we cannot communicate on (and simulating them is too costly). This has motivated the following alternative direction of~\cite{GhaffariP17_dfs}.

\medskip
\noindent
{\bf Ghaffari-Parter~\cite{GhaffariP17_dfs}.} In 2017, Ghaffari and Parter~\cite{GhaffariP17_dfs}
    provided the first (near-optimal, $\tilde O(D)$-rounds) distributed procedure that computes a path separator.
    Their procedure is {\em randomized} and overcomes the above challenges as follows.
    
    To overcome Challenge I, they compute for each node $f$ of $T^*$ a weight $w'(T^*_f)$ that approximates the total vertex-weight of the primal subgraph corresponding to $T^*_f$. Intuitively, instead of assigning each face the weight of all its vertices (which as mentioned above leads to overcounting),
    they sample each vertex with a certain probability. The weight of a face is 1 if at least one of the vertices of the face is sampled and it is 0 otherwise.\footnote{This assumes {\em unit} vertex-weights, however it can be extended to arbitrary weights by adjusting the sampling probability.}  Eventually, for each node $f$ in $T^*$ they compute the subtree OR of it (to check if at least one of the values in the subtree is 1). They repeat the process a poly-logarithmic number of times with different sampling rates, and show that this allows them to estimate the total vertex-weight of the corresponding primal subgraph. This allows them to overcome Challenge I. However, the approach is heavily based on randomization.
    
    To overcome Challenge II, they prove that if no balanced node exists in $T^*$, then a critical node $f$ must exist. 
   After finding such critical node $f$, a single artificial edge $e=(u,v)$ is added to the interior of the corresponding primal face $f$ in $G$, splitting $f$ into two faces (replacing the node $f$ in $T^*$ with two nodes connected by the new dual edge $e^*$).  
   Finally, they prove that one of the two new nodes is a balanced node, hence defining a cycle separator consisting of the $u$-to-$v$ path $P$ in $T$, and the single artificial edge $e\not \in G$.
      
      Apart from being randomized, the only drawback of the Ghaffari-Parter algorithm is that it requires the input graph to be bi-connected (i.e., the removal of any single vertex does not disconnect the graph). Nonetheless, they show that the bi-connectivity assumption is not a restriction for the specific application of computing a Depth-First Search (DFS) tree in near-optimal $\tilde O(D)$ rounds (following an approach given in the  parallel algorithm of~\cite{AggarwalA88_dfs}).

\medskip
\noindent
{\bf Li-Parter~\cite{LiP19_diameter}.}
     In 2019, Li and Parter~\cite{LiP19_diameter} showed how to entirely remove the bi-connectivity assumption of~\cite{GhaffariP17_dfs}, by augmenting $G$ with certain (non-triangulation) artificial edges that do not violate planarity and can be simulated efficiently in the distributed setting.
Another important contribution of~\cite{LiP19_diameter}, is showing how to apply the separator algorithm recursively, obtaining a {\em Bounded Diameter Decomposition (BDD)} in $\tilde{O}(D)$ rounds. 
The BDD is a recursive decomposition of planar graphs using separators.
It is the distributed analog of the centralized recursive separator decomposition. 
It is highly non-trivial since, in contrast to the centralized decomposition, the distributed decomposition is required to preserve a low $\tilde{O}(D)$ diameter for all subgraphs in the decomposition.
 
Since the BDD works with~\cite{GhaffariP17_dfs}'s separator, it is also randomized. Nevertheless, it extends the applications of the separator to include (randomized) distributed algorithms for classical problems such as  shortest-paths~\cite{LiP19_diameter}, maximum $st$-flow~\cite{AbdElhaleemDPW25_flow}, diameter~\cite{LiP19_diameter}, and reachability~\cite{Parter20_reachability}. See also \cref{sec:applications}.

\medskip
\noindent
{\bf Jauregui-Montealegre-Rapaport~\cite{deterministic_sep}.}
Very recently, in PODC '25, Jauregui, Montealegre and Rapaport showed how to obtain a near-optimal $\tilde O(D)$-rounds distributed {\em deterministic} separator algorithm, following a completely different approach than~\cite{GhaffariP17_dfs},
which does not use planar duality at all. Instead, they address Challenge I directly on the primal graph, and show how to approximate what they call  {\em fundamental face weights}, which are the total vertex-weight of the primal subgraphs enclosed by $C(T,e)$ for all $e\in G$. 
In our terminology, they approximate  the total vertex-weight of the primal graph corresponding to a dual subtree $T^*_f$ for all $f\in T^*$.

After computing the fundamental face weights, they check for one that can be used as a separator.
If one does not exist, then there is no edge $e\in G$ that closes a cycle in $T$ that can be used as the separator.
To handle this (Challenge II), they approximate the total vertex-weight of subgraphs enclosed by  fundamental cycles $C(T,e)$ defined by a certain set of what they call {\em augmentation edges} $e\notin G$.
They prove that one such edge $e$ exists that closes a cycle $C(T,e)=P\cup\{e\}$ in $G\cup\{e\}$, where $P$ is the desired separator path, and compute it.
The computational tasks they do on $T$ for both cases are not very complicated, however, the proofs of correctness are. Namely, they involve examining multiple cases of augmentation edges and $G$ edges, 
related to properties of augmentation edges, to DFS tours on $T$,  and to the embedding of $G$.

As an application of their algorithm, they show that a deterministic DFS algorithm  follows from using their separator in the DFS algorithm of Ghaffari-Parter.

\medskip
\noindent
{\bf Our approach.} 
We show that a simple approach where each vertex transfers its weight to a single face works also in the distributed setting and without triangulation. Since each vertex transfers its weight to exactly one face this overcomes the issue of overcounting (Challenge I).
It extremely simplifies upon the approach of~\cite{deterministic_sep} for computing face-weights.
Both~\cite{GhaffariP17_dfs} and~\cite{deterministic_sep} 
design other solutions that lead to  either complicated or random algorithms to deal with the two challenges above. 

We achieve the best of both worlds.
Our algorithm is a combination of Lipton-Tarjan~\cite{LiptonT80_sep} and Ghaffari-Parter~\cite{GhaffariP17_dfs}, that allows us to exploit the power of planar duality, which significantly simplifies the algorithm and its correctness compared to~\cite{deterministic_sep}. Concretely, we deal with Challenge I exactly like Lipton-Tarjan (i.e. every vertex transfers its weight to an arbitrary face it lies on), and we deal with Challenge II exactly like Ghaffari-Parter (i.e. using the fact that if there is no balanced node then there must be a critical node). Our contribution is in showing that this simple approach works, and can be implemented distributively.

Our correctness proof is inspired by the proofs of Lipton-Tarjan and Ghaffari-Parter. However, as we cannot triangulate the graph, we need to carefully handle the case of a critical node. 
In addition, we handle vertex-weighted graphs and adapt the proof to our weight assignment to faces. 
Finally, the distributed implementation is similar to both~\cite{GhaffariP17_dfs, deterministic_sep} and uses very standard tools in distributed computing such as low-congestion shortcuts and part-wise aggregation.

\section{A Deterministic Separator Algorithm}
\label{sec: separator}
Our algorithm begins by transferring weights from vertices to faces using a simple procedure: Each vertex transfers its weight to one arbitrary face it belongs to. 
  The weight of the face is then the total weight of vertices that transferred their weight to this face. 
  Then, we find a node $f$ in the dual tree $T^*$ that is either a balanced node or a critical node (where the weights are now on nodes of $T^*$). We prove that the fundamental cycle separator defined by $f$ is also a separator  w.r.t. the original vertex-weights. 
  A pseudocode of the algorithm is given in \cref{algorithm: separator}. The algorithm mostly follows the (distributed) algorithm of~\cite{GhaffariP17_dfs}. 
  The highlighted steps are the steps where we differ from~\cite{GhaffariP17_dfs}. For each such step we provide a proof of correctness. 
  We assume that the input graph is {\em bi-connected}. 
  This assumption can be removed as in~\cite{LiP19_diameter}. We discuss this in more detail in the next section. 

\begin{algorithm}[htb]
    \caption{Separator}
    \label{algorithm: separator}
    
    \KwIn{A bi-connected $D$-diameter graph $G$, a spanning tree tree $T$ of $G$, and a
    $\frac{1}{12}$-proper
    weight assignment $\wV(\cdot)$ to the vertices of $G$.}
    \KwOut{A 
    $3/4$-balanced 
    path separator $P$ of $G$ and an edge $e\notin T$ s.t. $P\cup\{e\}=C(T,e)$. \hspace{1in} (If $e\notin G$, then, adding $e$ to $G$ preserves planarity).}
    
    \medskip 
    \begin{enumerate}[leftmargin=*,align=left]
    
    \item \label{alg_step: learn_cotree}

    Compute the cotree $T^*$\;
    
    \item \label{alg_step: face_weights}

    \hl{Transfer vertex-weights to face-weights. I.e., to nodes of $T^*$  (\cref{observation: vertex_to_face_weight})}\;
   
    \item \label{alg_step: balanced_critical_node}
    Detect a \hl{$(\frac{1}{4},\frac{3}{4})$}-balanced dual node or a \hl{$(\frac{1}{4},\frac{3}{4})$}-critical dual node
    \hl{(\cref{lemma: balanced_critical_node})}\;

   \item \label{alg_step: mark_separator} 
   Mark the separator path $P$ in $T$, and learn vertices $u,v\in G$ s.t. $S=P\cup \{e\}$, where  \\ $e=(u,v)$ is  possibly  not in $G$ \hl{(\cref{lemma: separator_existence})} :
   \begin{enumerate}
        \item 
        If a balanced dual node $f$ was found, mark $P$ and $e=(u,v)$ where $P$ is the $u$-\\to-$v$ path in $T$ and $e^*=(f,\parent(f))$ in $T^*$ \;
        \item 
        If a critical dual node $f$ was found, find vertices $u,v$ on the face $f$ s.t. if $e=(u,v)$ \\ is added  to $f$  it creates a balanced node in $T^*$. Mark the $u$-to-$v$ path $P$ in $T$ as be-\\fore.
    \end{enumerate}
   \end{enumerate} 
\end{algorithm}

\begin{observation}
    \label{observation: vertex_to_face_weight}
    Let $G$ be a planar graph with vertex-weights $\wV (\cdot)$, and $S$ a cycle in $G$. Then, the weight assignment where each vertex transfers its weight to an arbitrary face it belongs to  
    satisfies $\wV(S_{in}^-) \leq \wF(\Sin) \leq \wV(\Sin)$. 
    Moreover, $\wF(G) = \wV(G)$. 
\end{observation}
\begin{proof}
As each vertex transfers its weight to exactly one face, we clearly have $\wF(G) = \wV(G)$. 
In addition, if $S$ is a cycle, we have $\wV(S_{in}^-) \leq \wF(\Sin) \leq \wV(\Sin)$. The reason is that all vertices that are in the strict interior of $S$ transfer their weight to a face in $\Sin$ as all the faces that contain them are in $\Sin$, hence $\wV(S_{in}^-) \leq \wF(\Sin)$. Vertices that are part of the cycle $S$, can belong both to faces that are in $\Sin$ and to faces that are in $\Sout$, so they may transfer their weight to a face in $\Sin$ or not. Vertices that are in the strict exterior of $S$ are only contained in faces outside $S$ and do not transfer their weight to a face in $\Sin$. So overall we have $\wV(S_{in}^-) \leq \wF(\Sin) \leq \wV(\Sin)$. 	
\end{proof}

This simple observation  is the core of our algorithm, and the place where we differ  from the previous distributed solutions that deal with Challenge I (see  \cref{sec:overview}) in an either randomized or complicated way. 
Next, we prove that we can indeed find a fundamental cycle separator w.r.t. vertex-weights using this observation. 

Our separator would easily follow if one finds an $(\alpha,\beta)$-balanced dual node, as in Lipton-Tarjan, perhaps with different constants $\alpha,\beta$.
However, since we do not triangulate the graph, an $(\alpha,\beta)$-balanced node does not necessarily exist. If that is the case, then there must exist an $(\alpha,\beta)$-critical node as shown by~\cite{GhaffariP17_dfs}. 

\begin{restatable}[\cite{GhaffariP17_dfs}]{lemma}{LemmaBalancedCriticalNode} 
\label{lemma: balanced_critical_node} Let $G$ be a planar graph with face-weights $\wF(\cdot)$, $T^*$ a spanning tree of $G^*$ rooted at an arbitrary node, and constants $\alpha<\beta<1$.
Then, $T^*$ either contains an $(\alpha,\beta)$-balanced node, or an $(\alpha,\beta)$-critical node.
\end{restatable}
 \begin{proof}   If an $(\alpha,\beta)$-balanced node does not exist, then there is no node $f$ satisfying $\alpha\cdot w(T^*)\leq w(T^*_f)\leq\beta \cdot w(T^*)$. Hence, all nodes $f$ satisfy either $\alpha\cdot w(T^*)> w(T^*_f)$ or $w(T^*_f)>\beta\cdot w(T^*)$.
    We claim that the node $f$  furthest from the root (breaking ties arbitrarily) s.t. $w(T^*_f) > \beta\cdot w(T^*)$ is an $(\alpha,\beta)$-critical node, for otherwise it has a child $h$ of subtree of weight  $w(T^*_h) \geq \alpha\cdot w(T^*)$. Then, $\alpha\cdot w(T^*) \leq w(T^*_h) \leq \beta \cdot w(T^*)$, meaning $h$ is an $(\alpha,\beta)$-balanced node, contradicting the assumption.
\end{proof}

Now, we show that a 
$3/4$-balanced 
separator exists w.r.t. our weight assignment, by detecting either a balanced or a critical node. If a balanced node $f$ exists, then the edge $e^*=(f,\parent(f))$ in $T^*$ is the edge (dual to the edge $e$) that closes a cycle with $T$ forming the desired fundamental cycle separator.
Otherwise, a critical node $f$ exists. Intuitively, this node has a large subtree but its children have small subtrees, so none of them can define a balanced separator, but a consecutive subset of them can. In other words, we add artificial edges that are not in $T$ to the interior of $f$, creating multiple new nodes in $T^*$ resulting from partitioning $f$ to smaller faces. Those nodes are connected to the children of $f$ such that one of them is balanced. The edge $e$ connecting this balanced node to its parent, determines which children of $f$ are in which side of the separator.
By proving that, we prove that there always exist two vertices $u,v$ s.t. the $u$-to-$v$ path $P$ in $T$ is a balanced path separator of $G$. Then, adding $e$ (whether it exists in $G$ or not), we have that $S=P\cup\{e\}$ is a fundamental cycle separator of $G\cup\{e\}$.

A similar statement was proven by~\cite{GhaffariP17_dfs}, but for their way of dealing with the face weight challenge (Challenge I in \cref{sec:overview}). We show a proof that works with our weight assignment to faces.
In addition,~\cite{GhaffariP17_dfs} focused on unit vertex-weights and we allow arbitrary vertex-weights. 
 
\begin{lemma}
\label{lemma: separator_existence}
Let $G$ be  a bi-connected planar graph, $T$ a spanning tree of $G$, and a $\frac{1}{12}$-proper 
weight assignment $\wV(\cdot)$ to the vertices of $G$. After transferring vertex-weights $\wV(\cdot)$ to face-weights $\wF(\cdot)$ as in \cref{observation: vertex_to_face_weight}, then:

\begin{enumerate}
    \item 
    If a $(\frac{1}{4},\frac{3}{4})$-balanced 
    node $f$ in $T^*$ exists w.r.t. $\wF(\cdot)$, let $u,v$ be the endpoints of the edge $e$ that is dual to the edge $(f,\parent(f))$ in $T^*$. Then, the $u$-to-$v$ path $P$ in $T$ is a $\frac{3}{4}$-balanced 
    separator of $G$ w.r.t. $\wV(\cdot)$.  I.e., $P\cup\{e\}=C(T,e)$. 
    
\item 
    Otherwise, there exists a $(\frac{1}{4},\frac{3}{4})$-critical 
    node $f$ in $T^*$ w.r.t.~$\wF(\cdot)$. Then, there are two vertices $u,v$ on $f$ s.t. the $u$-to-$v$ path $P$ in $T$ is a $\frac{3}{4}$-balanced 
    separator of $G$ w.r.t.~$\wV(\cdot)$. Moreover, if we add the edge $e=(u,v)$ to $G$, it does not violate planarity, and $P\cup\{e\}=C(T,e)$. 
    
\end{enumerate}    
\end{lemma}

\begin{proof}
By \cref{lemma: balanced_critical_node}, either a $(\frac{1}{4},\frac{3}{4})$-balanced node or a $(\frac{1}{4},\frac{3}{4})$-critical node exists  w.r.t. $\wF(\cdot)$.

\medskip 
\noindent  {\bf Case 1}: If a  
$(\frac{1}{4},\frac{3}{4})$-balanced node $f$ in $T^*$ exists w.r.t. $\wF(\cdot)$, then consider the fundamental cut $\delta(T^*,e^*)$, where $e^*=(f,\parent(f))$. 
The cut is dual to a fundamental cycle $S=C(T,e)$ in $G$, such that its interior corresponds to $T^*_f$ (\cref{fact: face_weight_duality}). See \cref{fig: cycle_cut} for an illustration.
We next prove that the total vertex weight of the strict interior and strict exterior of $S$ are both bounded by 
$\frac{3}{4}\cdot\wV(G)$. Then, setting $P$ to be $S\setminus\{e\}$ gives the balanced path separator.
    By \cref{fact: face_weight_duality}, $w(T^*_f)=\wF(\Sin)$. Also, as $f$ is 
    $(\frac{1}{4},\frac{3}{4})$-balanced we have:
   $\frac{1}{4}\cdot \wF(G)\leq w(T^*_f)=\wF(\Sin) \leq \frac{3}{4}\cdot \wF(G)$, 
    and since $\wV(S_{in}^-) \leq \wF(\Sin) \leq \wV(\Sin)$ (\cref{observation: vertex_to_face_weight}),
    we get that 
$\wV(S_{in}^-) \leq \wF(\Sin)\ \leq \frac{3}{4}\cdot \wF(G)$
 and   
   $\frac{1}{4}\cdot \wF(G)\leq \wV(\Sin).$
    Finally, because $\wF(G)=\wV(G)$ (\cref{observation: vertex_to_face_weight}) we get that both the strict interior is bounded by 
    $\wV(S_{in}^-) \leq \frac{3}{4}\cdot \wV(G)$
    and the strict exterior is bounded as we have 
    $\frac{1}{4}\cdot \wV(G)\leq \wV(\Sin),$
    which implies that 
    $\wV(\Sout^-)= \wV(G)-\wV(\Sin)\leq  \wV(G) - \frac{1}{4}\cdot \wV(G)=\frac{3}{4}\cdot\wV(G)$.
    Hence, we get a 
    $\frac{3}{4}$-balanced separator with respect to vertex-weights. 

\medskip
\noindent {\bf Case 2}: 
If a  
$(\frac{1}{4},\frac{3}{4})$-critical
node $f$ exists in $T^*$  w.r.t. $\wF(\cdot)$, then by the bi-connectivity assumption, $f$ is a simple cycle in $G$ \footnote{I.e, if not, then $f$ consists of a simple cycle and trees hanging from it. Then, removing a single vertex from the tree disconnects the graph, contradicting the assumption that it is bi-connected.}.  
    If $f$ has no children in $T^*$, then this cycle has only one of its edges $e=(f,\parent(f))$ in $T^*$. The rest of its edges form a path $P$ in $T$ (by \cref{fact: tree_cotree_decomposition}).  
 In this case, $P$ is the balanced separator, because the strict interior of the face (cycle) $f$ is empty (so weighs zero), and the strict exterior is of weight at most 
 $\frac{3}{4}\cdot \wV(G)$. This is because as $f$ is a critical node with no children in $T^*$, we have that 
 $w(T^*_f) = \wF(f) > \frac{3}{4} \cdot \wF(G) = \frac{3}{4} \cdot  \wV(G)$, where the last equality follows from \cref{observation: vertex_to_face_weight}. Now, if the weight of $f$ is larger than 
 $\frac{3}{4}\cdot  \wV(G)$ it means that the total weight of vertices in $f$ is at least 
 $\frac{3}{4} \cdot  \wV(G)$, because the weight of $f$ is the total weight of vertices in $f$ that transferred their weight to $f$. This means that the strict exterior of $f$ weighs 
 $\wV(G\setminus f)=\wV(G)-\wV(f)\leq \wV(G) - \frac{3}{4}\cdot\wV(G) < \frac{3}{4} \cdot\wV(G)$.

    Otherwise, $f$ has at least one child, and each child's  subtree weighs less than 
    $\frac{1}{4}\cdot \wF(G)$. We show that there are two vertices $u,v$ on $f$ s.t. the $u$-to-$v$ path $P$ in $T$ is a balanced separator for $G$ w.r.t. $\wV(\cdot)$. We do that by reducing this case to Case 1. 
    In other words, we show that if one would triangulate $f$ with artificial edges (see \cref{subfig: case2_augmented_graph}), then one of the resulting faces is 
    $(\frac{1}{4},\frac{3}{4})$-balanced. By Case 1, this implies that such two nodes $u,v$ exist. In this case, only the edge $(u,v)$ will be artificial.
   
    \begin{figure}[htb]
    \centering
    \begin{subfigure}{0.31\textwidth}
    \includegraphics[width=\textwidth]{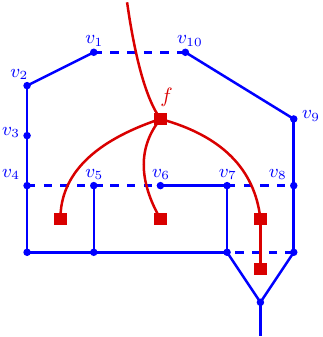} 
    \caption{Input graph \label{subfig: case2_input_graph}}
    \end{subfigure}
    \hspace{0.5 cm}
    \begin{subfigure}{0.27\textwidth}
    \includegraphics[width=\textwidth]{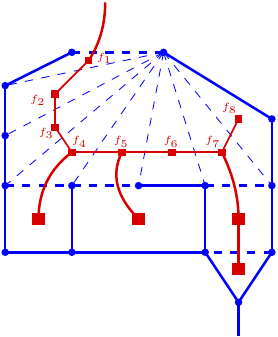}
    \caption{Augmented graph \label{subfig: case2_augmented_graph}}
    \end{subfigure}
    \hspace{0.5 cm}
    \begin{subfigure}{0.27\textwidth}
    \includegraphics[width=\textwidth]{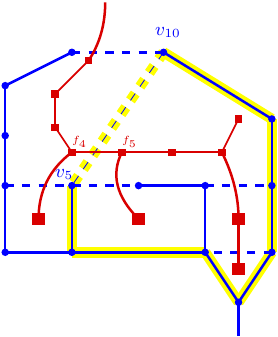}
    \caption{Separator \label{subfig: case2_separator}}
    \end{subfigure}
    \caption{\label{fig:case2}
    The primal graph is blue. The dual tree $T^*$ is red. The solid blue edges are edges of the primal tree $T$. Dashed edges are not in $T$ (i.e., their duals are in $T^*$).
    (a) $f$ is the face with vertices $v_1,v_2,\ldots,v_{10}$.
    The edge $e=(v_1,v_{10})$ corresponds to the dual edge $e^*=(f,\parent(f))$ in $T^*$.
    (b) Adding the artificial triangulation edges (thin dashed) to the interior of $f$ partitions $f$ into faces $f_1,f_2,\ldots, f_{8}$. In the new dual tree $T^*$, node $f_{i+1}$ is the child of $f_{i}$. The path $P(f)$ is the path from $f_1$ to $f_{8}$.
    (c) In this example, $f_j = f_4$.
    The dual edge $e^*=(f_4,f_5)$ in $T^*$ corresponds to the artificial primal edge $e=(v_5,v_{10})$ and defines the separator (highlighted in yellow).
    }
\end{figure}

    Formally, consider the primal edge $e$ whose dual is the edge $(f,\parent(f))$ in $T^*$. Let $v_1,\ldots,v_k$ be the vertices of $f$ such that $e=(v_1,v_k)$, and $f$'s vertices are enumerated counter-clockwise from $v_1$ to $v_k$. See \cref{subfig: case2_input_graph}.      
    For the sake of the proof, we triangulate $f$ by adding artificial edges 
    between $v_{k}$ and vertices $v_2,\ldots, v_{k-2}$ of $f$. This creates $k-2$ new faces $f_1,f_2,\ldots,f_{k-2}$, where $f_i$ is the face that contains the edge $(v_{i},v_{i+1})$ of $f$ (see \cref{subfig: case2_augmented_graph}).
    Adding the triangulation edges changes $T^*$ as follows: (1) The node $f$ is replaced by a $f_1$-to-$f_{k-2}$ path $P(f)$ of artificial edges, 
    (2) The parent of $f$ is now the parent of $f_1$, (3) Every $f_i$ (where $i<k-2$) has one (artificial) child $f_{i+1}$ and perhaps one more (real) child $g$, if $g$ was a child of $f$ in $T^*$ (i.e., when $(f,g)$ is the dual edge of $(v_i,v_{i+1})$), (4) $f_{k-2}$ is the last vertex of $P(f)$ so it does not have an artificial child, but it may have (at most) two real children (connected to $f_{k-2}$ via the edges dual to $(v_{k-2},v_{k-1})$ and $(v_{k-1},v_{k})$). 

        Consider $T^*$, where each node is assigned a weight as in \cref{observation: vertex_to_face_weight}. Recall that each vertex $v$ transfers its weight to one face. If this face is not $f$, then $v$ transfers its weight to the same face as before. If $v$ originally transferred its weight to $f$, then $v$ is contained in at least one face $f_i$. We then have $v$ transfer its weight to one such (arbitrary) face $f_i$. By this weight assignment the total weight of the new faces $f_i$ is exactly equal to the original weight of the face $f$. Since the original vertex weights were
    $\frac{1}{12}$-proper 
    and each $f_i$ has size 3 (i.e., gets the weight of at most 3 vertices), we have the following.
    \begin{claim}
    \label{claim: proper_face_weights}
        The weight of every $f_i$ is  
        $\frac{1}{4}$-proper.
    \end{claim}

We next prove that one of the dual nodes $f_i$ is balanced.

    \begin{claim}\label{claim:fi}
    Let $f_j$ be the lowest node on $P(f)$ whose subtree weighs more than  
    $\frac{3}{4}\cdot\wF(T^*)$
    but its children in $T^*$ have subtrees weighing at most 
    $\frac{3}{4}\cdot\wF(T^*)$.
       The child  $f_{j+1}$ of $f_j$ is
       $(\frac{1}{4},\frac{3}{4})$-balanced. 
    \end{claim} 
    \begin{proof}[Proof of \cref{claim:fi}.]   
    First we prove that such $f_j$ exists. Recall that $f$ was chosen such that it is a 
    $(\frac{1}{4},\frac{3}{4})$-critical node.
    This means that 
    $w(T^*_f) > \frac{3}{4} \cdot w(T^*)$,
    but all its children in $T$ have subtrees of weight smaller than 
    $\frac{1}{4}\cdot w(T^*)$. After we replace the face $f$ with the path $P(f)$, the total weight of all nodes $f_i$ in $P(f)$ is the same as the original weight of $f$, and the weight of all other nodes is the same. Therefore, we have that 
    $w(T^*_{f_1}) > \frac{3}{4}\cdot w(T^*)$. On the other hand, consider the last node $f_{k-2}$ on $P(f)$. We show that 
    $w(T^*_{f_{k-2}}) < \frac{3}{4} w(T^*)$. 
    For this, first note that 
    each face $f_i$ has three vertices, which means that in the dual graph the node $f_i$ has at most three neighbors. One of the neighbors of $f_{k-2}$ is the node $f_{k-3}$ of $P(f)$ (the parent of $f_{k-2}$ in $T^*$) and it has at most two children not in $P(f)$. Since $f$ is a critical node, each one of these neighbors $g$  not in  $P(f)$ has 
    $w(T^*_g) < \frac{1}{4}\cdot w(T^*)$, as explained above. Also, as the weights of the nodes $f_i$ are  
    $\frac{1}{4}$-proper (\cref{claim: proper_face_weights}), we get that $w(T^*_{f_{k-2}}) < \frac{1}{4} \cdot w(T^*) + 2 \cdot \frac{1}{4} \cdot w(T^*) = \frac{3}{4} \cdot w(T^*)$.

    Hence, there exists a node $f_j$ in $P(f)$ whose subtree weighs more than  
    $\frac{3}{4}\cdot\wF(T^*)$
    but its children in $T^*$ have subtrees weighing at most 
    $\frac{3}{4}\cdot\wF(T^*)$
    (note that the children not on $P(f)$ weigh less than 
    $\frac{1}{4}\cdot w(T^*)$ as explained above, and we are looking for the first $f_j$ in $P(f)$ where this condition holds). 
    Note that $f_j \neq f_{k-2}$ as we showed that
    $w(T^*_{f_{k-2}}) < \frac{3}{4} \cdot w(T^*)$. Hence, $f_j$ has a child $f_{j+1}$.
    By definition, we have that 
    $w(T^*_{f_{j+1}}) \leq \frac{3}{4} \cdot w(T^*)$.
    To show that $f_{j+1}$ is balanced we will show that 
    $w(T^*_{f_{j+1}}) > \frac{1}{4}\cdot  w(T^*)$.
    Recall that  
    $w(T^*_{f_j}) > \frac{3}{4}\cdot w(T^*)$.
    Recall, $f_j$ has at most 
    one child $g$ not on $P(f)$, in addition to the child $f_{j+1}$ on $P(f)$. The child $g$ not on $P(f)$ has 
    $w(T^*_g)<\frac{1}{4} \cdot w(T^*)$
    as explained above. And since the weights are $\frac{1}{4}$-proper we have that $\wF(f_j) \leq \frac{1}{4} \cdot w(T^*)$. On the other hand 
    $w(T^*_{f_{j+1}}) = w(T^*_{f_j}) - \wF(f_j) - w(T^*_g) > \frac{3}{4} \cdot w(T^*) - \frac{1}{4} \cdot w(T^*) - \frac{1}{4} \cdot w(T^*) = \frac{1}{4} \cdot w(T^*)$.
    Hence, $f_{j+1}$ is indeed a 
    $(\frac{1}{4},\frac{3}{4})$-balanced
    node as needed.
    \end{proof}
    
    By the above claim, Case 2 indeed reduces to Case 1. 
    That is, the edge $e=(f_j,f_{j+1})$ in the new dual tree defines two nodes $u,v\in G$ s.t. the $u$-to-$v$ path $P$ in $T$ is a 
    $3/4$-balanced
    path separator (in the augmented graph). See \cref{subfig: case2_separator}.
    It is not difficult to see that this is also the case after removing the artificial edges. 

\begin{restatable}{claim}{BalancedSepInputGraph} 
    \label{claim: balanced_sep}
             The $u$-to-$v$ path $P$ in $T$ is a  
             $3/4$-balanced
             path separator for the input graph (without artificial edges).
    \end{restatable}
 \begin{proof}  First, in the augmented graph with artificial edges we are exactly in Case 1: we find a balanced node, which translates to a $3/4$-balanced separator with respect to vertex-weights, as we proved in Case 1. This separator defines a cycle also in the original input graph (with possibly one additional edge). The set of vertices inside and outside this cycle are exactly the same in the augmented and original graph (and they have the same weight), and they are not affected by the artificial edges. Hence, the same separator is also a $3/4$-balanced separator in the original input graph. 	
 \end{proof}

Finally, $P$ and $e$ (if added to $G$) constitute a fundamental cycle $C(T,e)$ in $G\cup \{e\}$, and obviously $e$ does not violate planarity as it was embedded inside a face. 
    This concludes the proof of \cref{lemma: separator_existence}.
    \qedhere
    \end{proof}

The above concludes the correctness of \cref{algorithm: separator} which we summarize in the following \cref{thm: det_separator_sequential}.
In \cref{section: distributed_implementation} we show how to implement it distributively thus proving \cref{thm: det_separator}. 

\begin{theorem}
\label{thm: det_separator_sequential}
    Let $G$ be a bi-connected embedded planar graph, 
    $T$ a spanning tree of $G$, and $w(\cdot)$ a $1/12$-proper 
    weight assignment to $G$'s vertices. 
    Then, \cref{algorithm: separator}  deterministically finds a $3/4$-balanced path separator $P$ of $G$, where $P$ is a path in $T$. Moreover, there exists 
    an edge $e\notin T$ (maybe even $e\notin G$), such that, $S=P\cup\{e\}$ is a $3/4$-balanced fundamental cycle separator of $G$ (after adding $e$ in case $e\notin G$, in which case, $e$ does not violate planarity).
\end{theorem}

\section{A Distributed Implementation}
\label{section: distributed_implementation}

In this section we present the distributed implementation of \cref{algorithm: separator}. We first define some notions and briefly overview a set of existing tools in distributed algorithms. 
\subsection{Distributed Tools}
\label{sec: distributed_tools}

The main computational task our algorithm performs is computing {\em aggregations}. This is by now a simple  and standard building block in $\CONGEST$ algorithms. The advantage of an algorithm that works only with aggregations is that, it can be easily extended to be applied on multiple subgraphs of $G$ simultaneously, which is useful for applications. We elaborate on this more in  \cref{sec:the_distributed_algorithm}.

\begin{definition}[Aggregate operator]
      Given a set of $b$-bit strings $\{x_1,\ldots,x_m\}$, an aggregate operator $\oplus$ over all $x_i$ is denoted $\bigoplus_i x_i$, and is defined as the result of repeatedly replacing any two strings $x_i, x_j$ with the $b$-bit string  $x_i \oplus x_j$, until a single string remains. We use simple commutative operators like {\em AND}, {\em OR}, {\em SUM}.
\end{definition}

\noindent
{\bf Low-congestion shortcuts and aggregations.} {\em Low-congestion shortcuts} is a well known tool, first introduced in~\cite{GhaffariH16_shortcuts}, and is used to communicate in multiple subgraphs of the input graph simultaneously, in particular, it allows us to solve the classic {\em part-wise aggregation (PA)} problem efficiently.

\begin{definition}
    [Part-wise aggregation (PA)]
    Let $\{G_i\}_{i=1}^{k}$ be a partition of $G$ into vertex-disjoint connected subgraphs, such that
    each $v\in G_i$ has an $\tilde{O}(1)$-bit string $x_v$. The part-wise aggregation problem asks that each vertex of $G_i$ knows the aggregate value $\bigoplus_{v\in G_i} x_v$.
\end{definition}

Since $G_i$s are given as an input, they might be of arbitrarily large diameter. This is the challenge that low-congestion shortcuts resolve.

\begin{definition}
    [Low-congestion shortcuts~\cite{GhaffariH16_shortcuts}]
    \label{def: shortcuts}
    Let $\{G_i\}_{i=1}^{k}$ be a partition of $G$ into vertex-disjoint connected subgraphs, a $(c, d)$-shortcut, is a set of subgraphs $\{H_i\}_{i=1}^{k}$ of $G$ such that the diameter of each $G_i\cup H_i$ is at most $d$ and each edge of $G$ participates in at most $c$ subgraphs $H_i$.
\end{definition} 
\begin{lemma}[Shortcuts in planar graphs~\cite{GhaffariH21_shortcuts,HaeuplerIZ21_shortcuts}]
\label{lemma: shortcuts}
    Planar graphs admit construction of low-congestion shortcuts with parameters $c,d=\tilde{O}(D)$ within $\tilde{O}(D)$ deterministic rounds.
\end{lemma}

  After constructing shortcuts for a given partition $\{G_i\}_{i=1}^{k}$, the following task is solved deterministically in $\tilde{O}(c+d)$ rounds~\cite{HaeuplerIZ21_shortcuts,GhaffariH21_shortcuts}:

\begin{restatable}[PA in planar graphs~\cite{GhaffariH21_shortcuts,HaeuplerIZ21_shortcuts}]
{lemma}{PAPlanarGraphs}
\label{lemma: PA_primal}
     The PA problem can be solved in the given partition deterministically in $\tilde{O}(D)$ rounds on a planar graph $G$.
\end{restatable}

The PA problem is a very powerful building block that is used in many state-of-the-art algorithms.
In our case, we will use the following procedure repeatedly.
\begin{lemma}[Lemma 16 of~\cite{GhaffariZ22_mincut}]
\label{lemma: root_ancestor_subtree}
    Let $T$ be a tree, and $r$ be a vertex of $T$. There is a deterministic algorithm that roots $T$ at $r$, and for each node $u\in T$ computes the subtree sum using $\tilde{O}(1)$ part-wise aggregations.
\end{lemma}

\noindent
{\bf Dual computations.}
We will also perform computations related to faces of $G$ and to the dual tree $T^*$.
To do that, we use procedures of~\cite{GhaffariP17_dfs} that rely on low-congestion shortcuts and PA on a related graph $\hat{G}$, which can be fully simulated in $\CONGEST$ via $G$ with a constant overhead in the round complexity. See also Sections 3, 4 and appendix A of~\cite{AbdElhaleemDPW25_flow}. 

\begin{lemma}[Learn faces, Lemma 4 of~\cite{GhaffariP17_dfs}]
\label{lemma: know_faces}
    In  $\tilde{O}(D)$ deterministic rounds, a vertex $v\in G$ learns the IDs of faces that contain it, and for each incident edge $e$, learns the IDs of its endpoints in $G^*$.
\end{lemma}

\begin{lemma}[Aggregations on faces, Corollary 5 of~\cite{GhaffariP17_dfs}]
\label{lemma: aggregates_on_faces}
    Assuming each vertex $v$ has an input $x_{v,f}$ for each face $f$ of $G$ that it participates in. Then, 
    in $\tilde{O}(D)$ deterministic rounds any  aggregate operator $\bigoplus$ is computed over the input of vertices of each face $f$ in $G$ s.t. each vertex $v\in f$ learns $\bigoplus_{u\in f}x_{u,f}$.
\end{lemma}

\begin{lemma}[Subtree sums on dual tree, Section 4.2.2 of~\cite{GhaffariP17_dfs}]
\label{lemma: dual_subtree_sums}
    Let $g$ be a face whose ID is known to all $v\in G$, and let $T^*$ be a spanning tree of the dual graph $G^*$, such that, each vertex $v\in G$ knows for each incident edge $e$ whether it is in $T^*$, and $v$ knows the two IDs of faces that $e$ participates in. In addition, for each face $f$, there is an input value $x_f$ that is known to all of its vertices.
    Then, within $\tilde{O}(D)$ deterministic rounds, $T^*$ is rooted in $g$ and subtree sums are computed on $T^*$ w.r.t. $g$ as the root. I.e., each vertex $v$ of a face $f$ knows the ID of the parent of $f$ in the rooted $T^*$, and $v$ knows the subtree sum $\sum_{f'\in T^*_f}x_{f'}$.
\end{lemma}

\noindent
{\bf Derandomized components.}
Procedures that use~\cite{GhaffariP17_dfs} (like \cref{lemma: know_faces}, \cref{lemma: aggregates_on_faces}, and \cref{lemma: dual_subtree_sums})
were originally randomized.
However, they can be easily derandomized. In particular, those procedures use a classic randomized connectivity algorithm~\cite{GhaffariH16_shortcuts}. But, by~\cite{AbdElhaleemDPW25_flow} (footnote 20 in Appendix A), these procedures can be  derandomized using a derandomization of~\cite{GhaffariZ22_mincut} for the same connectivity algorithm.
In more detail, the connectivity algorithm is Boruvka-like~\cite{Boruvka}. Initially, each vertex is its own component. Then, distinct components merge over incident edges in phases. In order to have only a small (polylogarithmic) number of phases, it is beneficial if the merges have a simple structure, like star-shape. One common way to achieve it is to use a fair coin flip. It is however well-known, as shown for example in~\cite{GhaffariZ22_mincut}, that this can be achieved deterministically using the classical Cole-Vishkin~\cite{ColeVishkin} coloring algorithm.
In addition, traditional methods of aggregating information on the low-congestion shortcuts were originally randomized~\cite{GhaffariH16_shortcuts}, but since then were derandomized~\cite{GhaffariH21_shortcuts,HaeuplerIZ21_shortcuts}. Those components are standard, and used in \cite{deterministic_sep}'s deterministic separator algorithm as well.

\subsection{The Algorithm}
\label{sec:the_distributed_algorithm}
 We now give a distributed deterministic implementation of \cref{algorithm: separator}. Recall that  \cref{algorithm: separator} mostly follows Ghaffari-Parter~\cite{GhaffariP17_dfs}, the  lines highlighted in blue in \cref{algorithm: separator} specify which components differ from Ghaffari-Parter.
 For each such component, we describe the intuition behind its implementation and provide a proof.
 Proof ideas are provided also for steps that we borrow from previous work (e.g. \cite{GhaffariP17_dfs} or works on low-congestion shortcuts~\cite{GhaffariH16_shortcuts, GhaffariH21_shortcuts, HaeuplerIZ21_shortcuts}) in Appendix~\ref{appendix: deferred_details_implementation}.
The computational tasks we use are mainly  the standard part-wise aggregation primitive.

It is highly beneficial to describe the algorithm in terms of part-wise aggregations. Then, the algorithm can be applied simultaneously in multiple subgraphs via low-congestion shortcuts. This is often used in the applications, where also the spanning tree $T$ of $G$ might be arbitrary, which (from our perspective) is given as input. Thus, we assume a given partition of $G$ into vertex-disjoint bi-connected subgraphs $\{G_i\}_{i=1}^k$, and for each $G_i$ an arbitrary spanning tree $T_i$, and describe the algorithm for a specific subgraph $G_i$ assuming we have low-congestion shortcuts constructed. This guarantees  the algorithm can  run on all $G_i$ simultaneously at the same round complexity.
If one only wants to compute a separator of $O(D)$ vertices in $G$, then we compute a BFS tree $T$ of depth $O(D)$ of $G$ in $O(D)$ rounds~\cite{peleg-book}, and apply the algorithm on $G$ with $T$ as an input. To simplify notation, in the remainder of this section we  use $G$ and $T$ instead of $G_i$ and $T_i$.

\medskip
\noindent
{\bf Learning the cotree \boldmath$T^*$ (Step~\ref{alg_step: learn_cotree}).}
This step is the same as that of~\cite{GhaffariP17_dfs}. Intuitively, the given spanning tree $T$ of $G$ defines a dual cotree $T^*$ s.t. an edge is in $T^*$ iff it is not in $T$ (\cref{fact: tree_cotree_decomposition}). Since vertices of $G$ know their incident $T$ edges, they also know their incident $T^*$ edges. The correspondence between primal edges and dual edges (IDs of faces that contain them) is known in $\tilde{O}(D)$ rounds by \cref{lemma: know_faces}.
\begin{restatable}[Learn cotree~\cite{GhaffariP17_dfs}]{lemma}{LearnCoTree}
\label{lemma: learn_cotree}
    In deterministic $\tilde{O}(D)$ rounds, each vertex $v\in G$  
    learns for each incident $e$ edge whether it is in $T^*$.
\end{restatable}

\begin{proof}
    We use \cref{lemma: know_faces} to assign each face $f$ of $G$ an ID. In the output format, for each face $f$, vertices $v$ on $f$ know the ID of $f$. In addition, $v$ learns for each incident edge $e$ the two IDs of the faces containing $e$ (i.e., the endpoints of $e$ in $G^*$).
     Given this information, learning the cotree $T^*$ is immediate since all edges that are not in $T$ are in $T^*$ (\cref{fact: tree_cotree_decomposition}). Thus, each vertex $v\in G$ knows for each incident edge if it is in $T^*$.
\end{proof}

\medskip
\noindent
{\bf Reducing vertex-weights to face-weights (Step~\ref{alg_step: face_weights}).}
This step simplifies upon the approaches of~\cite{GhaffariP17_dfs, deterministic_sep} for computing face-weights and is the core of our algorithm.
To compute the face weights, we compute a summation operator over the faces of $G$, and the result is then the faces' weights as in \cref{observation: vertex_to_face_weight}. The input of a vertex $v$ for a face $f$ that contains it is simply $v$'s weight if $v$ chose to transfer its weight to $f$ and zero otherwise (e.g., $v$ transfers its weight to the minimal ID face that contains it).

\begin{lemma}[Assign face weights]
    \label{lemma: assign_face_weights}
    Given a weight assignment $\wV(\cdot)$ to the vertices of $G$, in deterministic $\tilde{O}(D)$ rounds, each face $f$ of $G$ is assigned a weight $\wF(f)$ equivalent to the sum of weights of vertices that transfer their weight to $f$ (as in \cref{observation: vertex_to_face_weight}). 
    Upon termination, all vertices of a face $f$ learn $\wF(f)$.
\end{lemma}
\begin{proof}
We want to assign each node $f$ in $T^*$ a weight as in \cref{observation: vertex_to_face_weight}, that is, 
each vertex $v$ that participates in faces $f_1,f_2,\ldots,f_k$ chooses an arbitrary face $f_i$ to transfer its weight to. Then, the weight $\wF(f)$ of a face $f$ is the total weight of $\wV(v)$ for vertices $v$ that choose to transfer their weight to $f$. 

To compute $\wF(f)$ for all faces $f$ in $G$, we use \cref{lemma: aggregates_on_faces}.
Note, in the input format, each vertex on a face $f$ has an input for $f$.
The input $v_f$ of a vertex $v$ for a face $f$ that contains it is, $\wV(v)$ iff $v$ chose $f$, and zero otherwise.
In order for vertices to be able to choose their inputs, they need to know the faces that contain them, which is done by~\cref{lemma: know_faces}. Hence, each $v\in G$ picks an arbitrary face (e.g. face with the minimal ID) of the faces that contains it to transfer its weight to it. Since now vertices know their inputs, via \cref{lemma: aggregates_on_faces} (with the aggregate operator SUM), we  compute in deterministic $\tilde{O}(D)$ rounds the weight of each face as in \cref{observation: vertex_to_face_weight}. 
Indeed, each vertex contributes to the weight of exactly one arbitrary face that contains it.
Moreover, in the output format of \cref{lemma: aggregates_on_faces}, all vertices on a face $f$ know the face's output $\wF(f)$.
\qedhere
    
\end{proof}

\medskip
\noindent
{\bf Balanced or critical node detection (Step~\ref{alg_step: balanced_critical_node}).}
Given the cotree $T^*$ (computed in Step~\ref{alg_step: learn_cotree}), either a $(\frac{1}{4},\frac{3}{4})$-balanced node or a $(\frac{1}{4},\frac{3}{4})$-critical node exists in $T^*$ (by \cref{lemma: balanced_critical_node}).
In order to implement this step, we mainly use \cref{lemma: dual_subtree_sums} to compute subtree sums on $T^*$, where the input for each node of $T^*$ is its assigned weight from the previous step.
Vertices of $G$ then detect locally a balanced (or critical) node that they participate in. Then, vertices elect some balanced or critical node by aggregating the IDs of these nodes on a BFS tree (or low-congestion shortcuts) of $G$ in deterministic $\tilde{O}(D)$ rounds by~\cite{peleg-book} (resp. by \cref{lemma: PA_primal}). 

\begin{restatable}{lemma}{DetectBalancedCriticalNode}
\label{lemma: balanced_critical_node_detection}
    There is an $\tilde{O}(D)$ deterministic round algorithm that finds a $(\frac{1}{4},\frac{3}{4})$-balanced dual node if exists in $T^*$, and a $(\frac{1}{4},\frac{3}{4})$-critical dual node otherwise.
    Upon termination, all vertices of $G$ learn the ID of the balanced or critical node.
\end{restatable}
\begin{proof}
Given the cotree $T^*$ (computed in Step~\ref{alg_step: learn_cotree}), either a $(\frac{1}{4},\frac{3}{4})$-balanced node or a $(\frac{1}{4},\frac{3}{4})$-critical node exists in $T^*$ (by \cref{lemma: balanced_critical_node}).
We find one of them as follows.

\begin{enumerate}
    \item 
    We want to compute the weight of $w(T^*_f)$ for all nodes $f\in T^*$. To do that, we want to use \cref{lemma: dual_subtree_sums}, but first we must provide its input format:
    (a) From \cref{lemma: assign_face_weights}, each $v\in G$ knows the weight of each face $f$ that contains $v$. In addition, (b) by \cref{lemma: know_faces} and \cref{lemma: learn_cotree}, each vertex $v$ knows for each incident edge $e$ whether it is in $T^*$, and $v$ knows the two IDs of faces that $e$ participates in. Finally, (c) we pick a root for $T^*$ by aggregating the IDs of faces in $G$,
    picking the maximum ID face (node of $T^*$) as the root. This is implemented in $\tilde{O}(D)$ deterministic rounds (\cref{lemma: PA_primal}).

    Henceforth, we match the input format of \cref{lemma: dual_subtree_sums}, so we use it to root $T^*$ and to compute the subtree sums of $T^*$ in deterministic $\tilde{O}(D)$ rounds. 
    Upon termination, all vertices of $v\in G$ on a face $f$, know the weight of $f$'s subtree $w(T^*_f)$.

    \item 
    To detect a balanced node or a critical node, vertices of $G$ on the (face corresponding to the) root node $g$ of $T^*$, broadcast the weight of $g$'s subtree (i.e., $w(T^*)$) to the entire graph $G$, then all vertices of $G$ learn $w(T^*)$.
    This is implemented in $\tilde{O}(D)$ rounds, either via the low-congestion shortcuts (\cref{lemma: shortcuts}) in case we are working in a subgraph $G_i$ of $G$, or via a BFS tree if we work in $G$~\cite{peleg-book}.

    \item
    Next, because each vertex $v\in G$ knows $w(T^*)$ and the subtree weight $w(T^*_f)$ of faces $f$ that $v$ participates in, $v$ locally detects if $f$ is a $(\frac{1}{4},\frac{3}{4})$-balanced node. All vertices aggregate on $G$, in $\tilde{O}(D)$ deterministic rounds (\cref{lemma: PA_primal}), the ID of balanced-nodes $f$ that they detected, choosing the maximum ID face $f$ to be the desired balanced node.

    \item 
    If no balanced node was identified in the previous step, there must be a $(\frac{1}{4},\frac{3}{4})$-critical node.
    Again, in the same manner as in the previous step, the ID of a critical node $f$ is known to vertices that participate in it and is learned by all vertices of $G$.\qedhere
\end{enumerate}
\end{proof}

\medskip
\noindent
{\bf Marking the separator (Step~\ref{alg_step: mark_separator}).}
In this step, we mark the $u$-to-$v$ path separator $P$ of $T$ and inform the vertices $u,v$ that they are the endpoints of the (possibly artificial) edge $e$ that closes a cycle with $P$.
We first detect the vertices $u,v$. Then, we mark the path $P$ by a simple subtree sum computation on $T$: $u,v$ have an input of 1, and other vertices 0. Thus, the edges that have one endpoint that has a sum of 1 and one endpoint that has a sum of at least one \footnote{I..e, all internal nodes of $P$, and only those nodes, have a sum that is exactly one; except for the lowest common ancestor of $u,v$ in $T$ that has a sum of two.}, are the edges of $P$. This terminates in $\tilde{O}(D)$ rounds (Lemmas~\ref{lemma: PA_primal} and~\ref{lemma: root_ancestor_subtree}).

The detection of $u,v$ is done by using a procedure of~\cite{GhaffariP17_dfs} (Appendices B.2 and B.3 of their paper) with different constants. In a high level, they investigate two cases (similar to what we do in \cref{thm: det_separator_sequential}). Namely, the case where a balanced node exists, and the case where a critical node exists, and for each they show how to deduce the vertices $u,v$.
Note, in \cref{lemma: separator_existence} we proved that given a balanced or a  critical node, such a separator exists using our way of assigning weights. The procedure uses simple primitives: subtree sums and previously computed weights (so, it can easily be made deterministic).
More details on the procedure are provided in Appendix~\ref{appendix: deferred_details_marking_sep_path}.

\begin{restatable}[Step S4 of~\cite{GhaffariP17_dfs}]{lemma}{MarkPathSeparator}
\label{lemma: mark_separator}
Assume that all vertices of a bi-connected graph $G$ know the ID of $(\frac{1}{4},\frac{3}{4})$-balanced node $f$ if exists, and otherwise, the ID of a $(\frac{1}{4},\frac{3}{4})$-critical node $f$.
In deterministic $\tilde{O}(D)$ rounds, each vertex $v\in G$ knows its incident edges on $P$ (if any), where $P$ is a $3/4$-balanced path separator of $G$. In addition, $P\cup\{e\}$ is a fundamental cycle $C(T,e)$ after adding the edge $e=(u,v)$, where $u,v$ are the endpoints of $P$. If $e\notin G$, then $u,v$ learn the ID of each other, and a local ordering (embedding) of $e$ that preserves planarity. 
\end{restatable}

\medskip
\noindent
{\bf Removing the bi-connectivity assumption (concluding \cref{thm: det_separator}).}
We remove the bi-connectivity assumption as in~\cite{LiP19_diameter}.
We mention that~\cite{LiP19_diameter}'s  procedure was originally randomized for using the connectivity algorithms of~\cite{GhaffariP17_dfs, GhaffariH16_shortcuts}. Hence, it is easily derandomized as in \cref{sec: distributed_tools}.

\begin{lemma}[Appendix B.1 of~\cite{LiP19_diameter}]
\label{lemma: biconnectivity_assumption_removal}
    Given a planar graph $G$, there is a set of artificial edges $E'$, such that, $G'=G\cup E'$ is a bi-connected planar graph.
    $G'$ is computed in deterministic $\tilde{O}(1)$ part-wise aggregation tasks.
    Moreover, any round of communication in $G'$ is simulated deterministically within $O(1)$ rounds in $G$. 
\end{lemma}

We use the above to conclude the proof of \cref{thm: det_separator}. That is, we compute $G'$ and run the algorithm on $G'$. Note that $T$ is indeed a spanning tree of $G'$ (as the algorithm expects). However, $T$'s edges are edges of $G$ (not artificial), and since the separator $P$ is a path of $T$, it is also a path of $G$. Obviously, $P$ is a balanced separator also for $G$ because there is no change in the vertex-weights. Hence:

\DistributedDeterministicSeparator*

\noindent
{\bf Generalization to multiple subgraphs.}
Note, given a partition  $\{G_i\}_{i=1}^{k}$ of $G$ into vertex-disjoint connected subgraphs, with spanning trees  $\{T_i\}_{i=1}^{k}$, where $T_i$ is an arbitrary spanning tree of $G_i$,  we have described the algorithm for computing a separator of a single subgraph $G_i$ in deterministic $\tilde{O}(D)$ rounds.  It easily extends to run on all subgraphs $G_i$ simultaneously in the same round complexity.
This follows from standard use of low-congestion shortcuts~\cite{GhaffariP17_dfs, GhaffariZ22_mincut,RozhonGHZL22_shortestpaths}. A proof is provided in Appendix~\ref{appendix: details_genrelaization_to_subgraphs}  

\begin{restatable}{theorem}{SpearatorMultipleSubgraphs}
\label{lemma: spearator_multiple_subgraphs}
Let $\{G_i\}_{i=1}^{k}$ be a partition of $G$ into vertex-disjoint connected subgraphs, and $\{T_i\}_{i=1}^{k}$ be spanning trees of $G_i$s. In addition let $\{w_i(\cdot)\}_{i=1}^{k}$ be $1/12$-proper weight assignments to $G_i$s' vertices.
Then, there is a deterministic $\tilde{O}(D)$-round  distributed algorithm that for each $G_i$ finds a $3/4$-balanced path separator $P_i$ w.r.t. $w_i(\cdot)$ of $O(\depth(T_i))$ vertices. In addition, $P_i\cup\{e\}$ is  a fundamental cycle of $T_i$, where $e=(u,v)$ and $u,v$ are the endpoints of $P$. In case $e\notin G_i$, then $u,v$ learn the ID of each other, and a local ordering (embedding) of $e$ that preserves planarity.
\end{restatable}

\section{Applications}\label{sec:applications}
Computing a balanced separator lies at the core of state-of-the-art distributed algorithms for many classical optimization problems in planar graphs. In most cases, the separator algorithm is their {\em only} randomized component.  
In particular, it was shown  in~\cite{deterministic_sep} that the (near-optimal) randomized algorithm of~\cite{GhaffariP17_dfs} for computing a DFS tree can be derandomized by replacing the randomized separator algorithm with the deterministic one of~\cite{deterministic_sep}. 
However, it was not mentioned in~\cite{deterministic_sep}, that 
a deterministic separator algorithm implies  
 a derandomization to a collection of other state-of-the-art algorithms that use the {\em Bounded Diameter Decomposition} (BDD) as their only 
 randomized component\footnote{Many of those algorithms use low-congestion shortcuts, that were originally randomized, but as mentioned in \cref{section: distributed_implementation} were since then derandomized.}. Recall from \cref{sec:overview} that the BDD is a recursive decomposition of planar graphs using separators, s.t. all subgraphs in the decomposition have a low $\tilde{O}(D)$ diameter.
The implementation of the BDD~\cite{LiP19_diameter} assumes a black-box fundamental cycle separator 
algorithm that runs in $\tilde{O}(D)$ rounds and applies it recursively. In case there is no such separator, it is sufficient that the algorithm  provides a path separator and (the endpoints and embedding) of a virtual edge, that if added, closes a cycle with the path separator. 
Replacing the black-box randomized separator algorithm in the BDD construction with ours (or with that of~\cite{deterministic_sep}) immediately derandomizes the BDD and hence the following state-of-the-art algorithms for classical problems in {\em directed} planar graphs:  

\begin{itemize} 
\item Computing $\tilde{O}(D)$-bit distance labels to every vertex (in a planar graph with positive and negative weights) in $\tilde{O}(D^2)$ rounds~\cite{LiP19_diameter} (s.t. the distance between any two vertices can be deduced by their labels alone), hence also single-source shortest-paths in $\tilde{O}(D^2)$ rounds.
\item Computing $\tilde{O}(D)$-bit reachability labels to every vertex in $\tilde{O}(D)$ rounds~\cite{Parter20_reachability} (s.t. the reachability between any two vertices can be deduced by their labels alone), hence also single-source reachability in $\tilde{O}(D)$ rounds.
\item Detecting strongly connected components in $\tilde{O}(D)$ rounds~\cite{Parter20_reachability}.
\item Maximum $st$-flow in $\tilde{O}(D^2)$ rounds~\cite{AbdElhaleemDPW25_flow}.
\item Global directed minimum cut in $\tilde{O}(D^2)$ rounds~\cite{AbdElhaleemDPW25_flow}.
\item Minimum weight cycle in $\tilde{O}(D^2)$ rounds~\cite{Parter20_reachability}.
\end{itemize}

\bibliographystyle{plain}
\bibliography{main}

\appendix 

\section{Deferred Details and Proofs  from \cref{sec:the_distributed_algorithm}}
\label{appendix: deferred_details_implementation}

\subsection{Missing Details from Step~\ref{alg_step: mark_separator} (Marking the Separator Path)}
\label{appendix: deferred_details_marking_sep_path}
In this step, we aim to mark the $u$-to-$v$ path separator $P$ of $T$ and inform the vertices $u,v$ that they are the endpoints of the (possibly artificial) edge $e$ which closes a cycle with $P$.
We mark a path separator in the case of a balanced node exactly as \cite{GhaffariP17_dfs} do: After detecting a balanced node $f$, we learn the primal edge $e=(u,v)$ corresponding to the edge $(f,\parent(f))$ of $T^*$ which defines the separator path $P$ as the $u$-to-$v$ path in $T$ (\cref{lemma: separator_existence}). Then, we mark $P$ by a subtree sum procedure in $T$ within $\tilde{O}(D)$ deterministic rounds (Lemmas~\ref{lemma: PA_primal} and~\ref{lemma: root_ancestor_subtree}), where $u,v$ have an input of 1, and other vertices 0. Thus, the edges that have one endpoint which got a sum of 1 are the edges of $P$.
If a critical node was found, this also reduces to the earlier case where we detect two vertices $u,v$ and mark the $u$-to-$v$ path $P$ in~$T$. 

However, detecting $u,v$ in this case is harder (as demonstrated in \cref{lemma: separator_existence}). 
Since the procedure we use is identical to that of~\cite{GhaffariP17_dfs}, we refer the reader to their paper for more details. However,
we highlight why \cite{GhaffariP17_dfs}'s procedure for marking the separator works also for our assignment in case of a critical node:
During their procedure, they perform a binary search for finding the artificial edge $e=(u,v)$ (that if added to the interior of $f$, closes a fundamental cycle balanced separator). The binary search in each phase examines a single (artificial) triangulation edge $e$, and computes the total weight of faces enclosed by the cycle $S=C(T,e)$, checking if it is balanced or not.

Note however, adding $e$ to $f$ creates two faces, each consists of a (distinct) path of $f$ and the edge $e$. One of the new faces is internal to the cycle $S$, denoted $f_{in}$ (the other face $f_{out}$ is external to $S$).
Computing the total weight of faces in $S$'s interior, reduces to summing up the weights (previously computed) of faces that are enclosed by $S$, with one exception, $f_{in}$. That is, all those faces are (real) faces of $G$, except that $S$ contains also the artificial face $f_{in}$, whose weight was not computed before because it simply did not exist until this stage where the artificial edge $e$ is examined. 
Nevertheless,~\cite{GhaffariP17_dfs} find a simple remedy for this issue, the weight of $f_{in}$ can be simply computed by performing a subtree sum on the path  $f_{in}$ (which they denote $(x_{i,1},\ldots,x_{i,j})$\footnote{In our terminology, this path is a $v_i$-to-$v_{k}$ path for some $i$, in Case 2 in the proof of \cref{lemma: separator_existence}.}) of $G$, where the input of a vertex $v$ (in their case) is a unit-weight ($v$'s original weight) if $v$ does not appear in another face that is contained in $S$, otherwise, $v$ is assigned an input of zero. This is done to not overcount $v$.
In our case, the input of a vertex $v$ is simply the weight that $v$ (originally) chose to contribute to $f$ (either $v$'s original weight or zero), that is, $v$ does not even have to check if it appears in another face enclosed by $S$ or not, because our assignment of faces does not suffer from overcounting.
We note that there is a very subtle point, where the endpoints of the artificial edge $e$, that previously chose to transfer their weights to $f$, can now choose to transfer their weight either to $f_{in}$ or to $f_{out}$, either choice is fine (as shown in \cref{lemma: separator_existence}).

\subsection{Proof of \cref{lemma: spearator_multiple_subgraphs} (Generalization to Multiple Subgraphs)}
\label{appendix: details_genrelaization_to_subgraphs}  

\SpearatorMultipleSubgraphs*
 
First, we give a brief overview why this generalization is possible.
The first issue that needs to be resolved is that the algorithm needs to compute part-wise aggregations inside each $G_i$. I.e., there is a given partition in each $G_i$ that we want to compute part-wise aggregations with respect to, and in all $G_i$s simultaneously. However, is not hard to see that a collection of partitions of $G_i$s imply a partition of $G$. Thus, we can use~\cref{lemma: PA_primal} to compute aggregate operators. The second issue is that we cannot perform computations in the dual graph $G^*$ straightforwardly because it is not the communication graph. However, a simple trick of~\cite{GhaffariP17_dfs} allows us to construct shortcuts for the faces of each $G_i$, which are the nodes of $G^*_i$. This allows to compute aggregate operators over faces of $G_i$, and to compute subtree sums on dual trees $T^*_i$. 

\begin{proof}
We concretely describe how one can generalize \cref{thm: det_separator} (the implementation of \cref{algorithm: separator}) to work in all subgraphs $G_i$, simultaneously, deterministically, and in $\tilde{O}(D)$ rounds. While phrasing the algorithm in terms of part-wise aggregations gets us half way through, we still have to clarify some details. We break this into two cases, the first concerns computations on primal graphs, and the second on dual graphs. 

\medskip
\noindent
{\bf  Computations on $G_i$s.}
This case is easy, we simply use the fact that with low-congestion shortcuts, we can implement an aggregation based algorithm on all $G_i$ simultaneously, within the same round complexity of $\tilde{O}(D)$ rounds (\cref{lemma: PA_primal}). However, the issue is that we actually want to compute part-wise aggregations inside each $G_i$, that is, for each $G_i$, we are given a partition  $G_{i,1}, G_{i,2},\ldots,G_{i,k_i}$ into  vertex-disjoint connected subgraphs, and we want to solve a part-wise aggregation task over the partitions $\{G_{i,j}\}_{j=1}^{k_i}$ of all graphs $G_i$ simultaneously. 
Nevertheless, this issue is well studied in the literature (e.g.~\cite{GhaffariP17_dfs, RozhonGHZL22_shortestpaths, GhaffariZ22_mincut}). In particular, one can think of the collection of second layer partitions as one partition $\{G_{i,j}: i\in [k], j\in[k_i]\}$ of $G$ and straightforwardly construct shortcuts on it.
This results with the desired $\tilde{O}(D)$ rounds to compute an aggregate operator in every $G_{i,j}$ (for all $i,j$).
This is possible because each vertex $v$ of $G$ knows in which $G_i$ it is contained, and $v$ knows in which part $G_{i,j}$ it is contained. Thus, the partition is (distributively) known to $G$'s vertices, so low-congestion shortcuts can be constructed and part-wise aggregations can be computed (as in \cref{lemma: PA_primal}). 
Note, it may be the case that different vertices in different parts are in different stages of the algorithm, the issue in this is that some vertices may want to change the partition and construct shortcuts again, while other vertices are not in that stage yet.  However, this can be easily solved by scheduling the update of partitions and reconstruction of shortcuts to start at specific rounds that are given by the algorithm. E.g. divide the time into intervals of $t=D \log^c n$ (for some constant $c$ of the algorithm) rounds, such that the partitions and shortcuts begin to reconstruct at the beginning of each interval. Then, aggregations are computed in the remaining time in the interval.

\medskip
\noindent
{\bf  Computations on $G^*_i$s.} Computations on the dual graphs $G_i^*$ (or their spanning trees $T^*_i$) is an issue because it is not clear how to  simulate a dual graph distributively. Indeed, this topic was tackled for the first time in~\cite{GhaffariP17_dfs} who provided a method for simulating aggregates on faces of $G_i$ and on dual spanning trees. 

We face this challenge exactly as~\cite{GhaffariP17_dfs}.
Namely, all computations we perform on the dual graphs $G^*_i$, especially the algorithms we use in Lemmas~\ref{lemma: know_faces},~\ref{lemma: aggregates_on_faces}, and~\ref{lemma: dual_subtree_sums}, boil down to aggregations on low-congestion shortcuts for the faces of $G_i$s. The technical challenge that~\cite{GhaffariP17_dfs} overcome in this natural approach, is that the set of faces of a graph is not vertex-disjoint. Thus, it does not meet the standard definition of low-congestion shortcuts (\cref{def: shortcuts}).
However, they show how to construct an auxiliary graph $\hat{G_i}$ deterministically in a constant number of rounds, such that, constructing low-congestion shortcuts for a specific partition of $\hat{G_i}$ is done in $\tilde{O}(D)$ rounds (the same complexity for constructing shortcuts inside $G$). This implies low-congestion shortcuts for the faces of $G_i$ with the same near-optimal parameters that planar graphs admit (Lemma 5 of~\cite{GhaffariP17_dfs}).
Our tasks then reduce to aggregate information on shortcuts in the primal graph, and this is done as discussed above for the primal. 

\medskip
\noindent
{\bf  Conclusion.}
In each given time, we have a partition of $G$ into $G_i$s, this partition implies two shortcut sets that are built on $G$ (depending on the steps of the algorithm), such that the first is used to compute part-wise aggregation tasks on $G_i$s simultaneously and the second one is for computations on the dual graphs $G^*_i$s (specifically, their spanning trees $T^*_i$s). 
Those shortcuts are constructed deterministically with near-optimal parameters $c,d=\tilde{O}(D)$ as in \cref{lemma: shortcuts}. In the algorithm, we move between them, such that, aggregations at any given time are computed in at most one of them. After which, the partitions and their shortcuts may be recomputed (according to the steps of the algorithm).
Thus, our algorithm  is applied in all subgraphs $\{G_i\}_{i=1}^k$ simultaneously deterministically within the same round complexity  $\tilde{O}(D)$ of applying it on a single graph $G_i$. \qedhere

\end{proof}

\end{document}